\documentclass[copyright,creativecommons]{eptcs}
\providecommand{\event}{ITRS 2012} % Name of the event you are submitting to
\usepackage{xcolor}
\usepackage{latexsym,amsmath,amssymb}
\usepackage{prooftree}
\usepackage{graphicx}
\usepackage{enumerate}
    \providecommand{\event}{Luca Paolini (Ed.): ITRS 2012} % adapt
    \providecommand{\firstpage}{1} % adapt
\newif\ifsubmit
\submittrue
\ifsubmit
\newcommand{\EZ}[1]{#1}
\newcommand{\EZComm}[1]{}
\newcommand{\PG}[1]{#1}
\newcommand{\PGComm}[1]{}

\newcommand{\DAComm}[1]{}
\else
\newcommand{\EZ}[1]{\textcolor{blue}{#1}}
\newcommand{\EZComm}[1]{{\scriptsize \textcolor{blue}{[Elena{:} #1]}}}
\newcommand{\PG}[1]{\textcolor{magenta}{#1}}
\newcommand{\PGComm}[1]{{\scriptsize \textcolor{magenta}{[Paola{:} #1]}}}

\newcommand{\DAComm}[1]{{\scriptsize \textcolor{green}{[D{:} #1]}}}
\fi

\newcommand{\SALTA}[1]{}

\newenvironment{proof}[1][Proof]{\begin{trivlist}
\item[\hskip \labelsep {\bfseries #1}]}{\end{trivlist}}

\newcommand{\HSep}{\hbox to \textwidth{\bf\hrulefill}}
\newcommand{\lab}[1]{{\scriptstyle{\textsc{(#1)}}}}
\newcommand{\Space}{\hskip 0.8em}
\newcommand{\NamedRule}[4]{\scriptstyle{\textsc{(#1)}}\
\displaystyle                  %  #1 = nome regola
\frac{#2}{#3}\           %  #2 = premesse (modo math) %  #3 = conseguenza (modo math)
#4     %  #4 = side conditions (modo math)
}

\newcommand{\Rule}[3]{{\scriptstyle{}}\
\displaystyle
\frac{#1}{#2}\           %  #1 = premesse (modo math) %  #2 = conseguenza (modo math)
#3     %  #3 = side conditions (modo math)
}

\newcommand{\refToFigure}[1]{Figure~\ref{fig:#1}}
\newcommand{\refToSection}[1]{Section~\ref{sect:#1}}

\newenvironment{grammatica}{$\begin{array}[t]{lcll}}{\end{array}$}
\newcommand{\produzione}[3]{#1&{:}{:}=&#2 & \mbox{{\small{#3}}}}
\newcommand{\terminale}[1]{{\text{\tt #1}}}

\newcommand{\te} {\mathit{t}} %metavariabile termini
\newcommand{\val}{\mathit{v}}%metavariabile valori
\newcommand{\x}{\mathit{x}} %(meta)variabile generica
\newcommand{\y}{\mathit{y}} %(meta)variabile generica
\newcommand{\z}{\mathit{z}} %(meta)variabile generica
\newcommand{\X}{\mathit{X}} %(meta)variabile generica
\newcommand{\Y}{\mathit{Y}} %(meta)variabile generica
\newcommand{\nval}{\mathit{n}}%metavariabile numerali
\newcommand{\tenv}{\Gamma}
\newcommand{\xtenv}{\Delta}

\newcommand{\error}{\mathit{error}}
\newcommand{\BinExp}[3]{#2\mathrel{#1}#3}
\newcommand{\SumExp}[2]{\BinExp{\SumKw}{#1}{#2}}

\newcommand{\SumKw}{\terminale{+}}
\newcommand{\LambdaExp}[2]{\lambda #1.#2}
\newcommand{\LambdaExpWithType}[3]{\lambda #1:#2.#3}
\newcommand{\LambdaCode}[3]{\lambda #1[#2].#3}

\newcommand{\AppExp}[2]{#1\ #2}
\newcommand{\Unbind}[2]{\langle#1\ \vert\ #2\rangle}
\newcommand{\Rebind}[2]{#1[#2]}

\newcommand{\tsubst}{\mathit{s}}
\newcommand{\tsubstN}{\mathit{r}}

\newcommand{\context}{{\cal E}}
\newcommand{\emptycontext}{[\,]}
\newcommand{\subst}{\sigma}

\newcommand{\ev}{\longrightarrow}

\newcommand{\extractTEnv}[1]{\EZ{\mathit{tenv}}(#1)}
\newcommand{\extractXEnv}[1]{\EZ{\mathit{xenv}}(#1)}
\newcommand{\extractNEnv}[1]{\mathit{Xenv}(#1)}

\newcommand{\FV}{\mathit{FV}}

\newcommand{\Int}{{\mathbb Z}}
\newcommand{\ApplySubst}[2]   {#1\{#2\}}

\newcommand{\CancelSubst}[2]   {#1_{\setminus #2}}
\newcommand{\dom}{\mathit{dom}}
\newcommand{\range}{\mathit{rng}}
\newcommand{\inContext}[1]{\context[#1]}

\newcommand{\intType}{\terminale{int}}
\newcommand{\codeType}[2]{\langle{#1}\rangle{#2}}
\newcommand{\T}{\mathit{T}}

\newcommand{\funType}[2]{#1 \rightarrow {#2}}

\newcommand{\TypeWithMod}[2] {\Unbind{#1}{#2} } 

\newcommand{\IsWFExp}[3]           {#1\vdash #2 : #3}

\newcommand{\Subst}[2]   {{#1[#2]}}

\newtheorem{theorem}{{\bf Theorem}}
\newtheorem{lemma}[theorem]{{\bf Lemma}}

\newtheorem{example}[theorem]{{\bf Example}}

\title{{Reconciling} positional and nominal binding\thanks{This work has been partially supported
by MIUR DISCO - Distribution, Interaction, Specification,
Composition for Object Systems, and MIUR CINA - .}\\
}

\author{Davide Ancona
\institute{DIBRIS, Univ. di Genova, Italy}
\and Paola Giannini
\institute{DISIT, Univ. Piemonte Orientale, Italy}
 \and Elena Zucca
 \institute{DIBRIS, Univ. di Genova, Italy}
 }
\def\titlerunning{Reconciling positional and nominal binding}
\def\authorrunning{Ancona, Giannini and Zucca}

\begin{document}
\maketitle

\begin{abstract}
We define an extension of the simply-typed lambda-calculus where two different binding mechanisms{,} \emph{by position} and \emph{by name}, nicely coexist.
In the former, as in standard lambda-calculus, the matching between parameter and argument is done on a
\emph{positional} basis, hence $\alpha$-equivalence holds,
whereas in the latter it is done on a \emph{nominal} basis. The two mechanisms also respectively correspond to static binding, where the existence and type compatibility of the argument are checked at compile-time, and 
dynamic binding, where they are checked at run-time. 
 \end{abstract}

%%%%%%%%%%%%%%%%%%%%%%%%%%%%%%%%%%%%%%%%%%%%%%%%%%%%%

\section{Introduction}
Two different binding mechanisms which are both widely applied in computer science are \emph{binding by position} and \emph{binding by name}.
In the former, matching is done on a
\emph{positional} basis, hence $\alpha$-equivalence holds, as demonstrated by the de Bruijn presentation of the lambda-calculus. This models parameter passing in most languages.
In the latter, matching is done on a \emph{nominal} basis, hence $\alpha$-equivalence does not hold, as in name-based parameter passing, method look-up in object-oriented languages, and synchronization in process calculi. 
Usually, identifiers which can be $\alpha$-renamed are called \emph{variables}, whereas \emph{names} cannot be $\alpha$-renamed (if not globally in a
program) {\cite{AnconaMoggi04, NanevskiPfenning05}}.
An analogous difference holds between tuples and records, as
recently discussed  {by Rytz and Odersky} \cite{RytzOdersky10}. 
The record notation has been extremely successful in object-oriented languages, whereas functional languages use prevalently tuples for non curried functions.
{The positional notation allows developers not to be constrained to a particular choice of names; from the point of view of clients, instead, the nominal notation can be better, 
since names are in general more suggestive}. However, in both cases {developers and clients have to} agree on some convention, either positional or nominal.

\EZ{The aim of this paper is to define a very simple and compact calculus which smoothly 
integrates positional  and nominal binding, providing a ``minimal'' unifying foundation for these two mechanisms, and to investigate the expressive power of their combination.}
Notably, we extend the simply typed lambda-calculus with two constructs.
\begin{itemize}
\item An \emph{unbound term} $\Unbind{\tsubstN}{\te}$, with $\tsubstN=\x_1\mapsto\X_1,\ldots,\x_m\mapsto\X_m$, is a value representing ``open
code''. That is, $\te$ may contain free occurrences of variables $\x_1,\ldots,\x_{\EZ{m}}$ to be dynamically bound, when code will be used, through the global nominal interface offered by names $\X_1,\ldots,\X_{\EZ{m}}$. {Each occurrence of $\x_1,\ldots,\x_{\EZ{m}}$ in $\tsubstN$ is called an \emph{unbinder}.} 
\item To be used, open code should be passed as argument to a \emph{rebinding lambda-abstraction} $\LambdaCode{\x}{\tsubst}{\te}$, with $\tsubst=\X_1\mapsto\te_1,\ldots,\X_m\mapsto\te_m$. 
This construct behaves like a standard lambda-abstraction. However, the argument, which is expected to be open code, is not used as it stands, but \emph{rebound} as specified by $\tsubst$, and if some rebinder is missing a dynamic error occurs.
\end{itemize}
For instance,  the application $(\LambdaCode{\z}{\X\mapsto 1,\Y\mapsto 2}{\z})\Unbind{\x\mapsto \X,\y\mapsto\Y}{\SumExp{\x}{\y}}$ reduces to $\SumExp{1}{2}$, while $(\LambdaCode{\z}{\X\mapsto 1}{\z})\Unbind{\x\mapsto \X,\y\mapsto\Y}{\SumExp{\x}{\y}}$ reduces to $\error$.

This proposal is based on our previous extension of lambda-calculus with unbind and rebind primitives \cite{DezaniEtAl10,DezaniEtAl11} {(of which \cite{DezaniEtAl09} is a preliminary version)} and indeed shares with this work the ability to express static and dynamic binding mechanisms within the same calculus. {A thorough comparison between the current calculus and the calculi of \cite{DezaniEtAl10,DezaniEtAl11} is presented in Section \ref{subS:comparison}.}

In the rest of this paper, we first provide the formal definition of an untyped version of the calculus (\refToSection{untyped}), then of a typed version with its type system (\refToSection{typed}), for which we prove a soundness result in \refToSection{sound}. In Section \ref{subS:comparison} we compare this calculus with our previous calculi and with \EZ{various} other calculi and examine the meta-programming features offered by the calculus.  Finally, in the Conclusion we discuss future work.

\section{{Untyped calculus}}\label{sect:untyped}
The syntax and {reduction rules} of the {untyped} calculus are given in
\refToFigure{calculus}. {We assume infinite sets of \emph{variables} $\x$ and \emph{names} $\X$.}
\begin{figure}[h]
\HSep
\begin{grammatica}
\produzione{\te}{\x\mid\nval\mid\SumExp{\te_1}{\te_2}\mid\LambdaExp{{\x}}{\te}\mid\AppExp{\te_1}{\te_2}\mid\Unbind{\tsubstN}{\te}\mid
\LambdaCode{\x}{\tsubst}{\te}\mid\error } {{term}}\\
%\produzione{\tenv}{\x_1{:}\T_1,\ldots,\x_m{:}\T_m}{{type context}}\\
\produzione{\tsubstN}{\x_1\mapsto\X_1,\ldots,\x_m\mapsto\X_m }{{unbinding map}}   \\ 
\produzione{\tsubst}{\X_1\mapsto\te_1,\ldots,\X_m\mapsto\te_m }{{rebinding map}}   \\ \\
\produzione{\val}{{\nval}\mid\LambdaExp{\x}{\te}\mid {\mid\LambdaCode{\x}{\tsubst}{\te}\mid\Unbind{\tsubstN}{\te}\ (\FV(\te){\subseteq}\dom(\tsubstN))}}{{value}}\\
\produzione{\context}{\emptycontext\mid\SumExp{\context}{\te}\mid\SumExp{\nval}{\context}\mid\AppExp{\context}{\te}\mid\AppExp{\val}{\,\context}} {{evaluation context}}\\
\produzione{{\subst}}{{\x_1\mapsto{\te_1},\ldots,\x_m\mapsto{\te_m}}}{{substitution}}
\end{grammatica}

\HSep
\[  
\begin{array}{lcl}
 {\SumExp{\nval_1}{\nval_2}\ev\nval} &
  {\mbox{if}\quad\tilde{\nval}=\tilde{\nval}_1+^\Int\tilde{\nval}_2} & \lab{Sum} \\
{\AppExp{(\LambdaExp{{x}}{\te})}{\val}\ev\ApplySubst{\te}{\x\mapsto\val}}{}{}
  &\quad&\lab{App}
  \\   
 \AppExp{(\LambdaCode{\x}{\tsubst}{\te}) }   
    {\Unbind{\tsubstN}{\te'} } \ev
 \ApplySubst{\te}{\x\mapsto \ApplySubst{\te'}{{\y\mapsto \tsubst(\tsubstN(\y)) \mid \y\in\dom(\tsubstN)} } }
  &\range(\tsubstN)\subseteq\dom(\tsubst)& \lab{{AppRebindOK}} 
  \\   \\
 \AppExp{(\LambdaCode{\x}{\tsubst}{\te}) }   
    {\Unbind{\tsubstN}{\te'} } \ev \error
  &\range(\tsubstN)\not\subseteq\dom(\tsubst)& \lab{{AppRebindERR}} \\
\end{array}
\]
\[
\begin{array}{l@{\quad\quad}l}
 \prooftree  \te\ev \te'\quad\quad \context\not=\emptycontext
  \justifies \inContext{\te}\ev\inContext{\te'} \using \lab{Cont}
\endprooftree &
\prooftree \te\ev \error \quad\quad \context\not=\emptycontext
  \justifies \inContext{\te}\ev\error \using \lab{ContError}
\endprooftree
\end{array}
\]
\caption{Syntax and reduction rules}\label{fig:calculus} \HSep
\end{figure}

Terms of the calculus are $\lambda$-calculus terms, {\em unbound terms}, {\em rebinding lambda-abstractions}, and a term representing {\em dynamic error}. We also include
integers with addition for concreteness. We use $\tsubstN$  for \emph{unbinding maps}, which are finite maps from variables to names, and
  $\tsubst$  for \emph{rebinding maps}, which are finite maps from names to terms. {Note that a standard lambda-abstraction is \emph{not} a special case of rebinding lambda-abstraction, that is, the term $\LambdaCode{\x}{\emptyset}{\te}$ behaves differently from $\LambdaExp{\x}{\te}$.}
  
\EZ{Note that in unbound terms we write, say, $\Unbind{\x\mapsto\X}{\AppExp{\y}{\x}}$, rather than directly $\langle\AppExp{\y}{\X}\rangle$, that is, differently from, e.g., \cite{Nanevski03}, names are not terms� of the underlying language but we keep an explicit mapping from variables into names.
This distinction  is a tradition in module calculi \cite{AZ-JFP02} and the main motivation is to keep separate the intra-module language, or \emph{core} language 
(in our paper, the language used to write code which can be unbound/rebound, which is here  
lambda-calculus for simplicity) from the inter-module language (constructs at the meta-level\footnote{\EZ{In this paper they are limited to  the unbind and rebind constructs but 
they could include, for instance, 
a renaming construct.}}) whose semantics can then be given independently from the core language. 
With our approach the inter-module language/meta-level can be built smoothly 
on top of the core language, without changing its syntax/semantics.
The inter-module language could be even applied to terms coming from different languages.
}

The operational semantics is described by {the reduction rules in \refToFigure{calculus}.} We denote
by $\tilde{\nval}$ the integer represented by the constant $\nval$, {by $+^\Int$ the sum of integers,} and by $\dom$ and $\range$ the domain and range of a map, respectively. The application of a substitution to a term,
$\ApplySubst{\te}{\subst}$, is defined, together with free
variables, in \refToFigure{subst}{, where we denote  by $\CancelSubst{\subst}{S}$ the substitution obtained from $\subst$ by removing variables in set $S$.} 
\begin{figure}[h]
\HSep

$\begin{array}{l}
\\
  \FV(\x)=\{\x\}\\
  \FV(\nval)=\emptyset\\
  \FV(\SumExp{\te_1}{\te_2})=\FV(\te_1)\cup\FV(\te_2)\\
  \FV(\LambdaExp{\x}{\te})=\FV(\te)\setminus\{\x\}\\
  \FV(\AppExp{\te_1}{\te_2})=\FV(\te_1)\cup\FV(\te_1)\\
  \FV(\Unbind{\tsubstN}{\te})=\FV(\te)\setminus\dom(\tsubstN)\\
 { \FV(\LambdaCode{\x}{\tsubst}{\te})=(\FV(\te)\setminus\{\x\})\cup\FV(\tsubst)}\\
  { \FV(\X_1\mapsto\te_1,\ldots,\X_m\mapsto\te_m)=\bigcup_{i\in 1..m}\FV(\te_i)}
  \\ \\
  \ApplySubst{\x}{\subst}={{\te}}\quad\quad \mbox{if}\ {\subst(\x)={\te}}\\
  \ApplySubst{\x}{\subst}=\x\quad\quad \mbox{if}\ \x\not\in\dom(\subst)\\
 \ApplySubst{\nval}{\subst}=\nval\\
  \ApplySubst{(\SumExp{\te_1}{\te_2})}{\subst}=\SumExp{\ApplySubst{\te_1}{\subst}}{\ApplySubst{\te_2}{\subst}}\\
  {\ApplySubst{({\LambdaExp{\x}{\te}})}{\subst}={\LambdaExp{\x}{\ApplySubst{\te}{\CancelSubst{\subst}{\{\x\}}}}}\quad\quad \mbox{if}\ \x\not\in\FV(\subst)}\\
  {\ApplySubst{(\AppExp{\te_1}{\te_2})}{\subst}=\AppExp{\ApplySubst{\te_1}{\subst}}{\ApplySubst{\te_2}{\subst}}}\\
  \ApplySubst{\Unbind{\tsubstN}{\te}}{\subst}=\Unbind{\tsubstN}{\ApplySubst{\te}{\CancelSubst{\subst}{\dom(\tsubstN)}}}\quad\quad \mbox{if}\ \dom(\tsubstN)\cap\FV(\subst)=\emptyset \\
  {\ApplySubst{(\LambdaCode{\x}{\tsubst}{\te})}{\subst}=\LambdaCode{\x}{\ApplySubst{\tsubst}{\subst}}
                                 {\ApplySubst{\te}{\CancelSubst{\subst}{\{\x\}}}}\quad\quad \mbox{if}\ \x\not\in\FV(\subst)
  }\\
  \ApplySubst{{(}\X_1\mapsto\te_1,\ldots,\X_m\mapsto\te_m{)}}{\subst}=\X_1\mapsto\ApplySubst{\te_1}{\subst},\ldots,\X_m\mapsto\ApplySubst{\te_m}{\subst}
\end{array}
$ \caption{Free variables and {application of substitution}}\label{fig:subst} \HSep
\end{figure}
Note that an
unbinder (that is, a variable occurrence in the {domain} of an unbinding map) behaves like a $\lambda$-binder:  for instance, in a term
of shape $\Unbind{\x\mapsto\X}{\te}$, the unbinder $\x$ introduces a local
scope, that is, binds free occurrences of $\x$ in $\te$.  Hence, a
substitution for $\x$ is not propagated inside $\te$. Moreover, a
condition
 which prevents
capture of free variables, similar to the $\lambda$-abstraction
case, is needed. For instance, the
term
$\te=\AppExp{(\LambdaExp{\y}{\Unbind{\x\mapsto\X}{\AppExp{\y}{\x}}})}{(\LambdaExp{\z}{\x})}$
is stuck, since the substitution $\ApplySubst{\Unbind{\x\mapsto\X}{\AppExp{\y}{\x}}}{\y \mapsto
\LambdaExp{\z}{\x}}$ is undefined, and therefore the term does not reduce to
$\Unbind{\x\mapsto\X}{\AppExp{(\LambdaExp{\z}{\x})}{\x}}$, which would be, indeed,
wrong. This condition is enforced by the definition of substitution where we require that the free variables of the substitution are disjoint from the domain of the unbinding map. However, as for the similar requirement for substitution applied to a lambda-abstraction we can always $\alpha$-rename the variables in the domain of the unbinding map (and consistently in the body of the unbound term) to meet the requirement {(we omit the obvious formal definition)}.\EZComm{forse metterla?} Consider the term
$\te'=\AppExp{(\LambdaExp{\y}{\Unbind{\x'\mapsto\X}{\AppExp{\y}{\x'}}})}{(\LambdaExp{\z}{\x})}$ which is $\te$ with the unbinder $\x$ renamed to $\x'$. The term $\te'$ is $\alpha$-equivalent to $\te$, and reduces (correctly) to $\Unbind{\x'\mapsto\X}{\AppExp{(\LambdaExp{\z}{\x})}{\x'}}$.

The rules of the operational semantics for sum and standard application are the usual ones. For application of a rebinding lambda-abstraction to an unbound term the variable $\x$ in the body of the lambda-abstraction is substituted with $\te'$ in which {each unbinder is} substituted with the term bound to the corresponding name in the rebinding. This application, however, evaluates to $\error$ in case the domain of the rebinding map is not a superset of the range of the unbinding map. {We write the side condition in both rules for clarity, even though it is redundant in the former.}

\begin{example}\label{ex:syntax}
This example shows that unbound terms can be {arguments of both standard and rebinding lambda-abstractions}. Consider the term $\te$ that follows
\[\AppExp {(\ \LambdaExp{{\y}} { {\AppExp{ (\LambdaCode{\z}{\X\mapsto 2 }{\z} )}\ {{\y}}  }
                                                  +
                                                  {\AppExp{ (\LambdaCode{\z}{\X\mapsto 3 }{\z} )}\ {{\y}}  }    
                                                }\  )} {\Unbind{\x\mapsto\X}{\x+1}}
\]
applying the rules of the operational semantics we get the following reduction:
\[
\begin{array}{rclr}
\te & \ev & { {\AppExp{ (\LambdaCode{\z}{\X\mapsto 2 }{\z} )} {\Unbind{\x\mapsto\X}{\x+1}}  }
                                                  +
                                                  {\AppExp{ (\LambdaCode{\z}{\X\mapsto 3 }{\z} )} {\Unbind{\x\mapsto\X}{\x+1}}  }    
                                                } & \lab{App}\\
 & \ev &     (2+1)+
                                                  {\AppExp{ (\LambdaCode{\z}{\X\mapsto 3 }{\z} )} {\Unbind{\x\mapsto\X}{\x+1}}  }    
                                                   & \lab{AppRebindOK} \\ 
& \ev &     3+{\AppExp{ (\LambdaCode{\z}{\X\mapsto 3 }{\z} )} {\Unbind{\x\mapsto\X}{\x+1}}  }    
                                                   & \lab{Sum}\\  
& \ev &     3+(3+1)  & \lab{AppRebindOK} \\  
& \ev &     3+4  &  \lab{Sum}\\ 
& \ev &     7 &  \lab{Sum}\\                              
\end{array}
\]
\end{example}

{From now on, we will use the abbreviation $\Rebind{\te}{\tsubst}$ for $\AppExp{(\LambdaCode{\x}{\tsubst}{\x})}{\te}$.}

\begin{example}\label{ex:dynBinding}
\newcommand{\ff}{\mathit{f}}
The classical example showing the difference between static and dynamic scoping:
\begin{verbatim}
let x=3 in
  let f=lambda y.x+y in
    let x=5 in
      f 1
\end{verbatim}
can be translated as follows:
\begin{enumerate}
\item $\AppExp{(\LambdaExp
                    {{\x}}
                    {\AppExp{(\LambdaExp
                              {\ff}
                              {\AppExp{(\LambdaExp
                                        {\x}{\AppExp{\ff}{1}})}
                                      {5}}
                            )} 
                            {(\LambdaExp{\y}{\x+\y})}
                    }
              )}
              {3}$ which reduces to $4$ ({\em static scoping}), or
\item {$\AppExp{(\LambdaExp
                    {\x}
                    {\AppExp{(\LambdaExp
                              {\ff}
                              {\AppExp{(\LambdaExp
                                        {\x}{\AppExp{{\Rebind{\ff}{\X\mapsto\x}}}{1}})}
                                        {5}}
                            )} 
                            {\Unbind{\x\mapsto\X}{\LambdaExp{\y}{\x+\y}}}
                    }
              )} 
              {3}$} which reduces to {$6$} ({\em dynamic scoping}). 
\end{enumerate}
\end{example}
\begin{example}\label{ex:meta-prog}
The following example shows some of the meta-programming features offered by the open code and the rebinding lambda-abstraction constructs.
$$
f = \LambdaExp{\x_1}
          {
            \LambdaExp{\x_2}
                      {\Unbind
                        {\y_1\mapsto\X,\y_2\mapsto\X}
                        {{\AppExp{(\Rebind{\x_1}{\X\mapsto\y_1})}{\Rebind{\x_2}{\X\mapsto\y_2}}
                        } }
                      }
          }
$$ 
$f$ is a function manipulating open code: it takes two open code fragments, with the same global nominal interface containing the sole name $\X$,
and, after rebinding both, it combines them by means of function application; finally, it unbinds the result so that the resulting nominal interface contains again
the sole name $\X$. The fact that the unbinding map is not injective means that the free variables of the two combined open code fragments will be finally
rebound to the same value (that is, the same value will be shared). 

For instance, ${\Rebind{(\AppExp{\AppExp{f}{\Unbind{\x\mapsto\X}{\LambdaExp{\y}{\y+\x}}}}{\Unbind{\x\mapsto\X}{\x}})}{\X\mapsto 1}}$ reduces to $2$.
\end{example}

\EZ{As the examples above show, the most useful construct in many cases is the application of a rebinding to an expression $\Rebind{\te}{\tsubst}$, which is a shortcut for $\AppExp{(\LambdaCode{\x}{\tsubst}{\x})}{\te}$. We prefer to take as primitive the rebinding lambda-abstraction $\LambdaCode{\x}{\tsubst}{\te}$ because in this way we also have, for free,  rebindings as first-class values (they are terms of shape $\LambdaCode{\x}{\tsubst}{\x}$), 
with a syntax which is a smooth extension of lambda calculus.}

\section{Typed calculus}\label{sect:typed}
The syntax {and} operational semantics of the typed calculus {are} given in \refToFigure{calculusTy}. 

\begin{figure}[h]
\HSep
\begin{grammatica}
\produzione{\te}{ \x\mid\nval\mid\SumExp{\te_1}{\te_2}\mid\LambdaExp{\x{:}\T}{\te}\mid\AppExp{\te_1}{\te_2}\mid\Unbind{\tsubstN}{\te}\mid
\LambdaCode{{\x{:}{\T}}}{\tsubst}{\te}\  } {{term}}\\
\produzione{\tsubstN}{\x_1{:}\T_1\mapsto\X_1,\ldots,\x_m{:}\T_m\mapsto\X_m }{{unbinding map}}   \\ 
\produzione{\tsubst}{\X_1{:}\T_1\mapsto\te_1,\ldots,\X_m{:}\T_m\mapsto\te_m }{{rebinding map}}   \\ \\
\produzione{\T}{\intType\mid\funType{\T_1}{\T_2} \mid\TypeWithMod{\xtenv}{\T} } {{type}}\\
%\produzione{\tau}{\intType\mid\funType{\T_1}{\T_2} }  {{ground type}}\\
\produzione{\tenv}{\x_1{:}\T_1,\ldots,\x_m{:}\T_m}{{context}}\\ 
\produzione{\xtenv}{\X_1{:}\T_1,\ldots,\X_m{:}\T_m}{{name context}}\\ \\
\produzione{\val}{{\nval}\mid\LambdaExp{\x{:}\T}{\te}{\mid\LambdaCode{\x{{:}\T}}{\tsubst}{\te}}\mid\Unbind{\tsubstN}{\te}\ (\FV(\te)\subset\dom(\tsubstN)) }{{value}}\\
\produzione{\context}{\emptycontext\mid\SumExp{\context}{\te}\mid\SumExp{\nval}{\context}\mid\AppExp{\context}{\te}\mid\AppExp{\val}{\,\context}} {{evaluation context}}\\
\end{grammatica}
\HSep
\[  
\begin{array}{c}
 \AppExp{(\LambdaCode{\x:\T}{\tsubst}{\te}) }   
    {\Unbind{\tsubstN}{\te'} } \ev
 \ApplySubst{\te}{\x\mapsto \ApplySubst{\te'} {{\y\mapsto \tsubst(\tsubstN(\y)) \mid \y\in\dom(\tsubstN) } } }\quad\quad\lab{{AppRebind}}
\end{array}
\]
\caption{Syntax of typed calculus, and modified {reduction rules}}\label{fig:calculusTy} \HSep
\end{figure}

In typed terms, as usual, \EZ{variables and names} (either in lambda-abstractions or maps) are decorated with types. {We assume that in an unbinding map two variables which are mapped in the same name are decorated with the same type, hence there is an implicit decoration for names as well.} 

Types are either ground types: integer {and} function types, or \EZ{unbound types, that is, types for open code, that needs the rebinding of some names. More precisely,} a term has type $\TypeWithMod{\X_1{:}\T_1,\ldots,\X_m{:}\T_m}{\T}$ if the term needs the rebinding of the names $\X_i$ ($1\leq i\leq m$)  to terms of type $\T_i$ ($1\leq i\leq m$) in order to produce a term of type $\T$. 
\EZ{A sequence $\X_1{:}\T_1,\ldots,\X_m{:}\T_m$ is called a 
\emph{type context} and is well-formed if, for each $i,j$ ($1\leq i,j \leq m$) , $\X_i=\X_j$ implies $\T_i=\T_j$, and analogously for \emph{contexts} $\x_1{:}\T_1,\ldots,\x_m{:}\T_m$}. Two name contexts are equal modulo permutation and repetitions of type assignments; consequently, the two types
$\TypeWithMod{\Y_1{:}\intType,\Y_2{:}\funType{\intType}{\intType}}{\intType}$ and $\TypeWithMod{\Y_2{:}\funType{\intType}{\intType},\Y_1{:}\intType,\Y_1{:}\intType}{\intType}$ 
are considered equal. 

The operational semantics of the untyped and typed versions of the language differs only in the {rules for} application with rebinding, and the fact that we do not have rule $\lab{ContError}$. In this case we do not check the correctness of the rebinding, that is that we have at least a rebinding for each name, and this is of the right type, since, as we will prove, the type system enforces this \EZ{property} statically.

{For instance, the untyped term $\te=\AppExp{(\LambdaCode{\x}{\Y\mapsto 3}{\SumExp{\x}{4}})}{\Unbind{\y\mapsto\Y}{\AppExp{\y}{2}}}$ {reduces} with rule $\lab{AppRebindOK}$ of \refToFigure{calculus} to $\SumExp{(\AppExp{3}{2})}{4}$ which is a stuck term. However, the term cannot be assigned a type, since in the typed version of the unbound term $\Unbind{\y{:}\funType{\intType}{\intType}\mapsto\Y}{\AppExp{\y}{2}}$ the variable $\y$ and therefore the name $\Y$ have type $\funType{\intType}{\intType}$, whereas in the typed version of  the rebinding lambda-abstraction, $\LambdaCode{\x{:}\TypeWithMod{\Y{:}\intType}{\intType}}{\Y{:}\intType\mapsto 3}{\SumExp{\x}{4}}$, the name $\Y$ \EZ{has type $\intType$}.}

The typing rules use the subtyping relation defined in \refToFigure{subtyping}. 

 \begin{figure}[t]
\HSep
  \begin{center}
$\begin{array}{c}
\\
\NamedRule{Sub-int}{}{\intType\leq\intType}{}
\Space
\NamedRule{Sub-arr}{\T_1'\leq\T_1\Space\T_2\leq\T_2'}{\funType{\T_1}{\T_2}\leq\funType{\T_1'}{\T_2'}}{}
\Space
\NamedRule{Sub-unbind}{\xtenv_2\leq\xtenv_1\Space\T_1\leq\T_2}{\TypeWithMod{\xtenv_1}{\T_1}\leq\TypeWithMod{\xtenv_2}{\T_2}}{}
\\[4ex]
\NamedRule{Sub-context}
{\T_1\leq\T'_1,\ldots,\T_m\leq\T'_m}
{\X_1{:}\T_1,\ldots,\X_{m+k}{:}\T_{m+k}\leq\X_1{:}\T'_1,\ldots,\X_m{:}\T'_m}{}

\end{array}$
\end{center}
\caption{Typed calculus: subtyping rules}\label{fig:subtyping} \HSep
\end{figure}

Subtyping rules for $\intType$ and arrow types are standard; as usual, for arrow types
subtyping is contravariant in the type of the formal parameter, and covariant in the returned type. 
A similar consideration applies to unbound types: an unbound term of type $\TypeWithMod{\xtenv_2}{\T_2}$ can be safely replaced by another term
of type $\TypeWithMod{\xtenv_1}{\T_1}$ if the requirements expressed by the corresponding \EZ{name} context $\xtenv_1$ are weaker than those
of $\xtenv_2$, and the type $\T_1$ of the term obtained after rebinding is a subtype of $\T_2$.

Finally, subtyping for \EZ{name} contexts \EZ{coincides with} the usual notion of width and depth subtyping for record types:
a \EZ{name} context $\xtenv_1$ is more specific than $\xtenv_2$  if it defines at least the same names defined by $\xtenv_2$,  
associated with types that are allowed to be subtypes of the corresponding types in $\xtenv_2$.

 \begin{figure}[t]
\HSep
  \begin{center}
$\begin{array}{c}
\NamedRule{T-Num}{}{ \IsWFExp{\Gamma}{\nval}{{\intType} } }{}\Space
\NamedRule{T-Var}{\Gamma(\x)=\T}{\IsWFExp{\Gamma}{\x}{\T}}{}\Space
\NamedRule{T-Sum}{\IsWFExp{\Gamma}{\te_1} {{\intType} }\Space\IsWFExp {\Gamma}{\te_2}{{\intType}} }{\IsWFExp{\Gamma}{\SumExp{\te_1}{\te_2}}{{\intType}}}
%\Space
%\NamedRule{T-Error}{}{\IsWFExp{\Gamma}{\error}{\T}}
 \\[4ex]
\NamedRule{T-Abs}{\IsWFExp{\Subst{\tenv}{\x{:}\T_1}}{\te}{\T_2}}{\IsWFExp{\tenv}{\LambdaExpWithType{\x}{\T_1}{\te}}{{\funType{\T_1}{\T_2}}}}{}\Space
\NamedRule{T-App}
{\IsWFExp{\Gamma}{\te_1}{{\funType{{\T_1}}{\T_2}}}\Space
\IsWFExp{\Gamma}{\te_2}{{\T_1'}}\Space \T_1'\leq\T_1} {\IsWFExp{\Gamma}{\AppExp{\te_1}{\te_2}}{\T_2}}{}\\[4ex]
%\NamedRule{T-AppReb}
%{\IsWFExp{\Gamma}{\te_1}{{\funType{{\TypeWithMod{\xtenv'}{\T'}}}{\T}}}\Space
%\IsWFExp{\Gamma}{\te_2}{{\TypeWithMod{\xtenv}{\T'}}}\Space \xtenv\subseteq\xtenv'} {\IsWFExp{\Gamma}{\AppExp{\te_1}{\te_2}}{\T}}{}
%\\[4ex]
\NamedRule{T-Unbind}{
\begin{array}{l}
  \IsWFExp{\Subst{\tenv}{\extractXEnv{\tsubstN}}}{\te}{\T} 
  \end{array}
}  
{\IsWFExp{\tenv}{\Unbind{\tsubstN}{\te}} { \TypeWithMod{\extractNEnv{\tsubstN}}{\T}}}\Space
\NamedRule{T-Rebind}  {\begin{array}{l}
                                        \IsWFExp{\Subst{\tenv}{\x{:}\T'} } {\te}{\T}\quad\quad
                                        \tsubst=\X_1{:}\T_1\mapsto\te_1,\ldots,\X_m{:}\T_m\mapsto\te_m \\
                                         \IsWFExp{\tenv}{\te_i}{\T_i}  \quad(1\leq i\leq m)
                                       \end{array}
                                      }  
{\IsWFExp{\tenv}{\LambdaCode{x{:}(\TypeWithMod{\extractNEnv{\tsubst}}{\T'})}{\,\tsubst\,}{\te}}
{{\funType{(\TypeWithMod{\extractNEnv{\tsubst}}{\T'}) }{\T}}  }  }\\
\end{array}$
\end{center}
\caption{Typed calculus: typing rules}\label{fig:typing} \HSep
\end{figure}

\EZ{In the typing rules (see~\refToFigure{typing}) we use} the following notations \EZ{for extracting a name context from an unbinding/rebinding map, extracting a context from an unbinding map, and updating a context, respectively}:
\begin{itemize}
\item $\extractNEnv{\X_1{:}\T_1\mapsto\te_1,\ldots,\X_m{:}\T_m\mapsto\te_m}=\extractNEnv{\x_1{:}\T_1\mapsto\X_1,\ldots,\x_m{:}\T_m\mapsto\X_m}=\X_1{:}\T_1,\ldots,\X_m{:}\T_m$
\item $\extractXEnv{\x_1{:}\T_1\mapsto\X_1,\ldots,\x_m{:}\T_m\mapsto\X_m}=\x_1{:}\T_1,\ldots,\x_m{:}\T_m$
and 
\item $\Subst{\tenv}{\tenv'}(\x)=\tenv'(\x)$ if $\x\in\dom(\tenv')$, $\tenv(\x)$ otherwise.
\end{itemize}

The rules are quite standard: \EZ{variables have their declared type,} integers and lambda-abstractions have types not needing rebindings, and the sum operator requires parameters of integer type.   
Rule $\lab{T-App}$ is standard: the type of the actual parameter must be a subtype of the type of the formal one. 
We have two rules for application, both require that the left term {has} a function type.
The first $\lab{T-App}$ is the standard  rule, in which the type of the actual parameter is equal to the one of the formal one. The second application rule, $\lab{T-AppReb}$, in case the argument  reduces to an unbound term, the type of the formal parameter of the rebinding lambda-abstraction to which the function reduce, may provide rebindings for more names than the ones needed. 

For an unbound term the body of the term must have type $\T$ in the current environment $\tenv$
updated by the environment $\extractXEnv{\tsubstN}$ where the unbound variables have the type specified in $\tsubstN$.
\EZ{Note that the rule can be applied only if,}
in the resulting unbound type, the \EZ{name} context $\extractNEnv{\tsubstN}$ extracted from $\tsubstN$ \EZ{is} well-formed, according to the definition given above.
For instance, the \EZ{name} context extracted from $\x_1{:}\intType\mapsto\X,\x_2{:}\funType{\intType}{\intType}\mapsto\X$ is not well-formed, 
since the name $\X$ is required to have the two different types $\intType$ and $\funType{\intType}{\intType}$ at the same time.

{Finally, {the type of the formal parameter of a rebinding lambda-abstraction specifies the types of the names that are in the rebinding $\tsubst$, and the type that the variables $\x$ must have in order to type the body of the lambda-abstraction.}}

Let $\T=\TypeWithMod{\X{:}\intType}{\intType}$, the typed term corresponding to the untyped term of Example \ref{ex:syntax} is:
\[
\AppExp {(\ \LambdaExp{\y{:}\T} { {\AppExp{ (\LambdaCode{\z{:}\T}{\X{:}{\intType}\mapsto 2  ) }{\z} )} {\y}  }
                                                  +
                                                  {\AppExp{ (\LambdaCode{\z{:}\T}{\X{:}{{\intType}}\mapsto 3  ) }{\z} )} {{\y}}  }    
                                                }\  )} {\Unbind{\x{:}{{{\intType}}}\mapsto\X}{\x+1}}
\]

\section{Soundness of the calculus}\label{sect:sound}

The type system is {\em safe} since types are preserved by reduction, {\em subject reduction property}, and closed terms are not stuck, {\em progress property}.

The proof of subject reduction relays on the inversion and context lemmas that follows.

\begin{lemma} [Inversion]\label{l:inversion}\

\begin{enumerate}
\item\label{il1} If $\IsWFExp{\Gamma}{\x}{\T}$, then {$\T=\Gamma(\x)$}.
\item\label{il2} If $\IsWFExp{\Gamma}{\nval}{\T}$, then {$\T={\intType}$}.
\item\label{il3} If $\IsWFExp{\Gamma}{\SumExp{\te_1}{\te_2}}{\T}$, then  $\T={\intType}$, $\IsWFExp{\Gamma}{\te_1}{{\intType}}$, and $\IsWFExp{\Gamma}{\te_2}{{\intType}}$.
\item\label{il4} If $\IsWFExp{\Gamma}{\LambdaExp{\x{:}\T_1}{\te}}{\T}$,
then for some $\T_2$ we have
$\T={(\funType{\T_1}{\T_2})}$, and $\IsWFExp{\Subst{\tenv}{\x{:}\T_1}}{\te}{\T_2}$.
\item\label{il5} If $\IsWFExp{\Gamma}{\AppExp{\te_1}{\te_2}}{\T}$, then
for some $\T_1$, and $\T'_1$ we have
$\IsWFExp{\Gamma}{\te_1}{{(\funType{\T_1}{\T})}}$, $\IsWFExp{\Gamma}{\te_2}{\T'_1}$, and $\T'_1\leq\T_1$.
\item\label{il6} If $\IsWFExp{\Gamma}{\Unbind{\tsubstN}{\te}}{\T}$, then for some $\T'$ we have 
$\T=\TypeWithMod{\extractNEnv{\tsubstN}}{\T'}$, and 
$\IsWFExp{\Subst{\tenv}{\extractTEnv{\tsubstN}}}{\te}{\T'}$.
\item\label{il7} If ${\IsWFExp{\tenv}{\LambdaCode{x{:}\T'}{ \X_1{:}\T_1\mapsto\te_1,\ldots,\X_m{:}\T_m\mapsto\te_m }{\te}}{\T }  }$, then for some $\T'_1$, and $\T'_2$, $\T=\funType{\T'}{\T'_2}$, $\T'=\TypeWithMod{\X_1{:}\T_1,\ldots,\X_m{:}\T_m}{T'_1}$, $\IsWFExp{\Subst{\tenv}{\x{:}\T'_1} } {\te}{\T'_2}$, and for all $i$, $1\leq i\leq m$, $\IsWFExp{\tenv}{\te_i}{\T_i} $.
\end{enumerate}
\end{lemma}
\begin{proof}
By case analysis on typing rules. 
\end{proof}

\begin{lemma}[Substitution]\label{l:subst}
If $\IsWFExp{\Subst{\tenv}{\x_1{:}\T_1,\dots,\x_m{:}\T_m}}{\te}{\T}$,  
$\IsWFExp{\tenv}{\te_i}{\T'_i}$, and $\T'_i\leq\T_i$ ($1\leq i\leq m$), then
\PG{$\IsWFExp{\tenv}{\ApplySubst{\te}{\x_1\mapsto\te_1,\ldots,\x_m\mapsto\te_m}}{\T'}$ where $\T'\leq\T$.}
\end{lemma}
\begin{proof}
By induction on {terms} $\te$. 
\end{proof}

\begin{lemma}[Context]\label{l:context}
Let $\IsWFExp{\tenv}{\inContext{\te}}{\T}$, then
\begin{itemize}
\item $\IsWFExp{\tenv}{\te}{\T'}$ for some $\T'$, and
\item \PG{for all $\te'$, if $\IsWFExp{\tenv}{\te'}{\T''}$, and $\T''\leq\T'$, then
$\IsWFExp{\tenv}{\inContext{\te'}}{\T'''}$ for $\T'''\leq\T$ .}
\end{itemize}
\end{lemma}
\begin{proof}
By induction on evaluation contexts {$\context{}$}. 
\end{proof}

\begin{theorem}[Subject Reduction]\label{theo:srt}
If $\IsWFExp{\tenv}{\te}{\T}$ and $\te\ev\te'$, then $\IsWFExp{\tenv}{\te'}{\T}$.
\end{theorem}
\begin{proof}
\PG{
By induction on reduction derivations. We
consider only the interesting rules.\\
If the applied rule is $\lab{App}$, then
\[{\AppExp{(\LambdaExp{\x{:}\T_1}{\te})}{\val}\ev\ApplySubst{\te}{\x\mapsto\val}}.\]
By hypothesis $\IsWFExp{\tenv}{\AppExp{(\LambdaExp{\x{:}\T_1}{\te})}{\val}}{\T}$.
From Lemma \ref{l:inversion}, cases (\ref{il5}) and (\ref{il4})  we have that there is $\T'_1$ such that
$\IsWFExp{\Gamma}{\LambdaExp{\x{:}\T_1}{\te}}{{(\funType{\T_1}{\T})}}$,
$\IsWFExp{\Gamma}{\val}{\T'_1}$, and $\T'_1\leq\T_1$.
Again by Lemma \ref{l:inversion}, case
(\ref{il4}) we have that $\IsWFExp{\Subst{\tenv}{\x{:}\T_1}}{\te}{\T}$. By Lemma \ref{l:subst} we derive that
$\IsWFExp{\tenv}{\ApplySubst{\te}{\x\mapsto\val}}{\T}$.
}

\PG{
If the applied rule is $\lab{AppRebind}$, then 
\[
 \AppExp{(\LambdaCode{\x:\T'}{\tsubst}{\te}) }   
    {\Unbind{\tsubstN}{\te'} } \ev
 \ApplySubst{\te}{\x\mapsto \ApplySubst{\te'} {{\y\mapsto \tsubst(\tsubstN(\y)) \mid \y\in\dom(\tsubstN) } } }
\]
From Lemma \ref{l:inversion}, cases (\ref{il5}) and (\ref{il7}) we have that there is $\T''$ such that
$\IsWFExp{\Gamma}{\LambdaCode{\x:\T'}{\tsubst}{\te}}{{(\funType{\T'}{\T})}}$,
$\IsWFExp{\Gamma}{\Unbind{\tsubstN}{\te'}}{\T''}$, and $\T''\leq\T'$.
Again by Lemma \ref{l:inversion}, case
(\ref{il7}) we have that $\T'=\TypeWithMod{\extractNEnv{\tsubst}}{\T'''}$ for some $\T'''$, and
 $\IsWFExp{\Subst{\tenv}{\x{:}\T'''} } {\te}{\T}$.\\
From $\IsWFExp{\Gamma}{\Unbind{\tsubstN}{\te'}}{\T''}$, Lemma \ref{l:inversion}, case (\ref{il6}), 
we have that $\T''=\TypeWithMod{\extractNEnv{\tsubstN}}{\T''_1}$ for some $\T''_1$.
Assume that 
\[
\tsubst=\X_1{:}\T_1\mapsto\te_1,\ldots,\X_n{:}\T_n\mapsto\te_n
\]
 so $\extractNEnv{\tsubst}=\X_1{:}\T_1,\ldots,\X_n{:}\T_n$, from the subtyping
rule $\lab{Sub-unbind}$ we have that
\begin{itemize}
\item $T''_1\leq\T'''$, 
\item $\extractNEnv{\tsubstN}=\X_1{:}\T'_1,\ldots,\X_m{:}\T'_m$, where $m\leq n$, and
\item $\T_i\leq\T'_i$ for all  $1\leq i\leq m$.
\end{itemize}
From Lemma \ref{l:inversion}, case (\ref{il7}), 
we have that for all $i$, $1\leq i\leq m$, $\IsWFExp{\tenv}{\te_i}{\T_i} $, and from $\IsWFExp{\Gamma}{\Unbind{\tsubstN}{\te'}}{\T''}$, Lemma \ref{l:inversion}, case (\ref{il6}), we get that
$\IsWFExp{\Subst{\tenv}{\x_1{:}\T'_1,\dots,\x_m{:}\T'_m}}  {\te'}{\T''_1}$. 
Let $\te''=\ApplySubst{\te}{\x_1\mapsto\te_1,\ldots,\x_m\mapsto\te_m}=\ApplySubst{\te'} {{\y\mapsto \tsubst(\tsubstN(\y)) \mid \y\in\dom(\tsubstN) } }$, from Lemma \ref{l:subst}, we derive that $\IsWFExp{\tenv}{\te''}{\T''_2} $ and $\T''_2\leq\T''_1$. Therefore,  again Lemma \ref{l:subst}, $T''_2\leq\T'''$ (derived from transitivity of $\leq$), and 
$\IsWFExp{\Subst{\tenv}{\x{:}\T'''} } {\te}{\T}$, imply that $\IsWFExp{\tenv}{ \ApplySubst{\te}{\x\mapsto\te''} }{\T''_3}$ for $\T''_3\leq\T$, which is what we wanted to prove.
}

If the applied rule is $\lab{Cont}$, $\te= \inContext{\te_1}$, $\te'=\inContext{\te'_1}$ and $\te_1\ev\te_1'$. By Lemma \ref{l:context}, and 
$\IsWFExp{\tenv}{\inContext{\te_1}}{\T}$  for some  $\T'$, $\IsWFExp{\tenv}{{\te_1}}{\T'}$. \PG{By induction hypothesis, $\IsWFExp{\tenv}{{\te'_1}}{\T'_1}$ with $\T'_1\leq\T'$. Therefore, by Lemma \ref{l:context}, $\IsWFExp{\tenv}{\inContext{\te'_1}}{\T''_1}$ with $\T''_1\leq\T$.}
\end{proof}

In order to show the Progress Theorem, we first state the Canonical {Forms} Lemma, and then a lemma asserting the
usual relation between type contexts and free variables (Lemma \ref{l:free}).

\begin{lemma}[Canonical Forms]\label{l:canonical}\

\begin{enumerate}
\item\label{cfl1} If $\IsWFExp{}{\val}{{\text{\tt int}}}$, then
$\val=\nval$.
\item\label{cfl2} If $\IsWFExp{}{\val}{\TypeWithMod{\xtenv}{{\T}}}$, then $\val=\Unbind{\tsubstN}{\te}$ for some $\tsubstN$ and $\te$.
\item\label{cfl3} If $\IsWFExp{}{\val}{{(\funType{\T}{\T'})}}$, then either
$\val=\LambdaExp{\x{:}\T}{\te}$, or 
$\val=\LambdaCode{\x{{:}\T}}{\tsubst}{\te}$ for some $\tsubst$ and $\te$.
\end{enumerate}
\end{lemma}
\begin{proof}
By case analysis on the shape of values.
\end{proof}

\begin{lemma}\label{l:free}
If $\IsWFExp{\tenv}{\te}{\T}$, then
$\FV(\te)\subseteq\dom(\tenv)$.
\end{lemma}
\begin{proof}
By induction on type derivations.
\end{proof}

\begin{theorem}[Progress]\label{theo:pt}
If $\IsWFExp{}{\te}{\T}$, then {either} $\te$ is a
value, or $\te\ev\te'$ for some $\te'$.
\end{theorem}
\begin{proof}
By induction on the typing derivation of $\IsWFExp{}{\te}{\T}$.

If $\IsWFExp{}{\te}{\T}$ and $\te$ is not a value,  then the last  applied rule
cannot be $\lab{T-Num}$, 
$\lab{T-Abs}$, $\lab{T-Unbind}$, or $\lab{T-Rebind}$. Moreover  the typing environment for the
expression is empty, {hence} by Lemma \ref{l:free}  the last applied rule  cannot be $\lab{T-Var}$.

If the last typing rule applied is $\lab{T-Sum}$, then
$\te=\SumExp{\te_1}{\te_2}$ and:
\[
\Rule{\IsWFExp{}{\te_1}{{\intType}}\Space\IsWFExp{}{\te_2}
{{\intType}}}{\IsWFExp{}{\SumExp{\te_1}{\te_2}}{{\intType}}}{}
\]
If $\te_1$ is not a value, {then,} by induction hypothesis,
$\te_1\ev\te_1'$. So by rule $\lab{Cont}$, with context
$\context=\SumExp{\emptycontext}{\te_2}$, we have
${\SumExp{\te_1}{\te_2}}\ev{\SumExp{\te'_1}{\te_2}}$. If $\te_1$
is a value, {then,} by Lemma \ref{l:canonical},  case (\ref{cfl1}), $\te_1={\nval_1}$. Now, if
$\te_2$ is not a value, {then,} by induction hypothesis, $\te_2\ev\te_2'$.
So by rule $\lab{Cont}$, with context
$\context=\SumExp{{\nval_1}}{\emptycontext}$, we get
${\SumExp{\te_1}{\te_2}}\ev{\SumExp{\te_1}{\te'_2}}$. Finally,
if $\te_2$ is a value, {then}  by Lemma \ref{l:canonical}, case (\ref{cfl1}),
$\te_2={\nval_2}$. Therefore rule $\lab{Sum}$ is applicable.

\PG{If the last applied rule is $\lab{T-App}$, then
$\te={\AppExp{\te_1}{\te_2}}$, therefore  for some $\T'$ and $\T''$:
 \[\Rule
{\IsWFExp{}{\te_1}{{(\funType{\T'}{\T})}}\Space \IsWFExp{}{\te_2}{\T''}\Space\T''\leq T'}
{\IsWFExp{}{\AppExp{\te_1}{\te_2}}{\T}}{}\]
 If $\te_1$ is not a
value, then, by induction hypothesis, $\te_1\ev\te_1'$. So
$\AppExp{\te_1}{\te_2}= \inContext{\te_1}$ with
$\context=\AppExp{\emptycontext}{\te_2}$, and by rule $\lab{Cont}$,
$\AppExp{\te_1}{\te_2}\ev\AppExp{\te'_1}{\te_2}$. If $\te_1$ is a
value $\val$, {and} $\te_2$ is not a value, then, by induction
hypothesis, $\te_2\ev\te_2'$. So $\AppExp{\te_1}{\te_2}=
\inContext{\te_2}$ with $\context=\AppExp\val{\emptycontext}$, and
by rule $\lab{Cont}$, $\AppExp{\val}{\te_2}\ev\AppExp{\val}{\te'_2}$. \\
If both $\te_1$ and $\te_2$ are values, then by Lemma \ref{l:canonical},
case (\ref{cfl3}), 
\begin{enumerate}
\item $\te_1=\LambdaExp{\x{:}\T'}{\te}$, or
\item $\LambdaCode{\x{{:}\T'}}{\tsubst}{\te}$.
\end{enumerate}
For case (1), rule $\lab{App}$ can be applied. For case (2), from  Lemma \ref{l:inversion}, case (\ref{il7}), $\T'=\TypeWithMod{\extractNEnv{\tsubst}}{T'''}$.
Let $\extractNEnv{\tsubst}=\X_1{:}\T_1,\ldots,\X_n{:}\T_n$,  since $\T''\leq T'$ from rule $\lab{Sub-unbind}$, $\T''=\TypeWithMod{\X_1{:}\T'_1,\ldots,\X_m{:}\T'_m}{T''_1}$, where $m\leq n$, and rule $\lab{{AppRebind}}$ is applicable.
}
\end{proof}

\section{Related Work}\label{subS:comparison}

\subsection{Comparisons with our  previous calculi}\label{subS:compOur}
This proposal is based on our previous extension of lambda-calculus with unbind and rebind primitives \cite{DezaniEtAl10,DezaniEtAl11} and indeed shares with this work the ability to express static and dynamic binding mechanisms within the same calculus. However, there are two {main} novelties. Firstly, the explicit distinction between variables and names allows us a cleaner and simpler treatment of $\alpha$-equivalence, which only holds for variables\footnote{We thank an anonymous referee of \cite{DezaniEtAl10} for pointing out this problem.}, as in process and module calculi. Secondly, we investigate \EZ{here} a different semantics where rebinding is more controlled, that is, can only be applied to terms which effectively need to be rebound. 

The two previous points are reflected in the difference in the unbinding and rebinding constructs. 
\EZComm{Ho cambiato il confronto facendolo sui calcoli non tipati mi sembrava pi\`u chiaro}
In \cite{DezaniEtAl10,DezaniEtAl11},
\begin{itemize}
\item {the unbinding construct had shape} $\Unbind{\x_1,\ldots,\x_m}{\te}$, {specifying} a set of unbinders, 
\item \EZ{correspondingly,} {the rebinding construct had shape $\Rebind{\te}{\tsubst}$, specifying that the variables in the domain of $\tsubst$ were rebound,} 
%was the construct that when $\te$ was $\Unbind{\tenv}{\te'}$, if all the variables unbound where rebound with values of the right type would do the substitution, otherwise a dynamic error was raised (as in the current calculus). 
\item {applying a rebinding to an unbound term had a behavior as in the current calculus, but rebinding could also be applied to terms not reducing to unbound terms. For instance, }
in
\EZ{\[\AppExp{\Rebind{(\LambdaExp{\y}{\y+\Unbind{\x}{\x}})}{\x\mapsto
1}}{\Unbind{\x}{\x+2}}
\]}
the term \EZ{$\Unbind{\x}{\x+2}$} is rebound inside the lambda-expression. To produce this semantics, rebinding maps {were} pushed, with reduction rules,  inside lambdas (and applications) and {remained} stuck on variables. They {were} then resolved when, via a standard application, the variable is substituted with an unbind construct. A term such as $\Rebind{(\SumExp{\Unbind{\x\mapsto\X}{\x}}{4})}{\X\mapsto 1}$ is stuck in the current calculus, whereas its analogous in the calculi of \cite{DezaniEtAl10,DezaniEtAl11} reduces to $5$. 
\end{itemize}  

In the calculus of the current paper, unbinding/rebinding are mediated by the use of names, and rebinding is done via application of rebinding lambda-abstractions $\LambdaCode{\x}{\tsubst}{\te}$ to unbound terms. The operational semantics of the rebinding construct $\Rebind{\te}{\tsubst}$ of \cite{DezaniEtAl10,DezaniEtAl11}, corresponds to the one of term $\AppExp{(\LambdaCode{\x}{\tsubst}{\x})}{\te}$ of the current calculus, if we take only well-typed terms. However,  e.g., the term  $\SumExp{\Unbind{\x\mapsto\X}{\x}}{4}$ is well typed in the calculus of \cite{DezaniEtAl10,DezaniEtAl11}, but not in the current one.

{Another important difference w.r.t. previous calculi is that,} as the abbreviation introduced at the end of Example \ref{ex:syntax} suggests, the term $\LambdaCode{\x}{\tsubst}{\x}$, which is a value, may be thought as the {rebinding} $\tsubst$. {That is, as a matter of fact, rebindings are first-class values.} In the previous calculi in \cite{DezaniEtAl10,DezaniEtAl11} this was not the case, as it is not in \cite{NanevskiEtAl08}, where {rebinding is} applied via the use of metavariables.

Comparing the type systems of previous calculi with the current one, we can notice that in \cite{DezaniEtAl11} soundness was proved for a call-by-name semantics and did not hold for call-by-value. The introduction of intersection types, in \cite{DezaniEtAl10},  allowed us the characterization of  terms that could be used both as values and in contexts providing unbindings, and the proof of soundness for a call-by-value semantics. The more restricted semantics of rebinding of the calculus of the current paper allows us the definition of a simpler type system that does not require intersection types, to prove soundness for the call-by-value evaluation strategy. 

\subsection{Comparisons with other calculi}\label{subS:compOther}
\EZ{\subsubsection{Dynamic binding}}
 \EZ{As we can see from Example \ref{ex:dynBinding}, we are able to model dynamic scoping, where identifiers are resolved w.r.t. their dynamic environments, and rebinding, where  identifiers are resolved w.r.t. their static environments, but additional primitives allow explicit modification of these environments. 
Classical references for dynamic scoping are  \cite{M98}, and \cite{Dami97alambda-calculus}, whereas  the $\lambda_{\rm marsh}$ calculus of
\cite{BiermanEtAl03a} supports
rebinding w.r.t. named contexts (not of individual
variables). 
Our semantics} corresponds more closely  to 
what happens in the calculus for dynamic binding of Nanevski, see \cite{Nanevski03}, and
in the contextual modal type theory of \cite{NanevskiEtAl08}. 
However, in \cite{Nanevski03}, there are two severe limitations: lambda-abstraction may not contain ``names'' 
(this \EZ{means} in our setting that it is not possible to unbind a variable in a lambda), and unbound terms may not
have free variables that may be unbound to names later. Both this limitations, and Nanevski says it, prevent 
gradual unbinding, and therefore the utility of the calculus for  metaprogramming. In contextual modal type theory, 
see \cite{NanevskiEtAl08},
there may not be occurrences of free variables in unbound terms (this was not a limitation of \cite{Nanevski03}),  
whereas this may happen in our calculi. Consider the term $\te={\Unbind{\y{:}{\intType}\mapsto\Y}{\SumExp{\y}{\x}}}$. In contextual modal type theory, there is no environment $\tenv$ in which this term is well-typed, due to the occurrence of the free $\x$ in an unbound term. In our calculi, {instead,} in an environment in which $\x$ has type ${\intType}$, the term is {well-typed}. {This allows an expressive power similar to ``unquote'', even though a precise comparison is matter of further investigation. Indeed, if the term is in the scope of a lambda, say $\LambdaExp{\x{:}{\intType}}{\cdots\te\cdots}$, applying the lambda to an integer, say $3$, replaces such integer in the term $\te$.}

\subsubsection{Modules}\label{sect:cms}
An unbound term $\Unbind{\tsubstN}{\te}$ resembles a module in the CMS calculus \cite{AZ-JFP02}, having just an output
unnamed component; in CMS the input components of a module (that is, the external components on which the module depends on)
are represented exactly by an unbinding map  $\tsubstN$, whereas output components (that is, the components defined in the module
that are available outside) are represented exactly by a rebinding map $\tsubst=\X_1\mapsto\te_1,\ldots,\X_m\mapsto\te_m$, where each  
$\te_i$ may contain free occurrences of variables (corresponding to unbinders).
Such variables represent input components that have to be
provided dynamically by other modules through nominal interfaces and suitable operators for combining modules.

For instance, the CMS term\footnote{A module in CMS can contain also local components, but for simplicity here we
consider just input and output components.} $M_1=[\x_1\mapsto\X_1, \x_2\mapsto\X_2; \Y_1\mapsto 1, \Y_2 \mapsto \x_1+\x_2]$
represents a module defining two output components $\Y_1$ and $\Y_2$, where the definition of $\Y_2$ depends on both
the input components $\X_1$ and $\X_2$. In the calculus we have presented here, such a module can be represented
by the term $t_1=\Unbind{\x_1\mapsto\X_1, \x_2\mapsto\X_2}{\LambdaCode{\z}{\Y_1\mapsto 1, \Y_2 \mapsto \x_1+\x_2}{\z}}$.
Similarly, the module $M_2=[\;\X_1\mapsto 1,\X_2\mapsto 2]$ without input components, is represented by the term
$t_2=\LambdaCode{\z}{\X_1\mapsto 1,\X_2\mapsto 2}{\z}$. 

Whereas in CMS linking of the two modules $M_1$ and $M_2$ can be expressed as a combination of primitive module operators
yielding the final module $M=[\ \Y_1\mapsto 1, \Y_2 \mapsto 1+2;]$, here linking is  expressed in terms
of application: $\AppExp{t_2}{t_1}$ reduces to the term $t=\LambdaCode{\z}{\Y_1\mapsto 1, \Y_2 \mapsto 1+2}{\z}$ which, indeed,
represents the module $M$. Finally, selection of a module component, as $M.\Y_2$, can be expressed again in terms
of application and an unbound term: $\AppExp{t}{\Unbind{\x\mapsto\Y_2}{\x}}$ reduces to $1+2$ which, in turns, reduces to $3$ as expected.

\subsubsection{Meta-programming}

We have already shown how the calculus supports meta-programming 
features to promote dynamic composition of software components (example 3 of \refToSection{untyped}, and
\refToSection{cms}). 

In particular, when components are composed together, it is possible to identify and/or
to discriminate components (dependening on the specific need) in a simple way.
Let us consider the following expression:
$$
f_1 = \LambdaExp{c_1}
          {
            \LambdaExp{c_2}
                      { \Unbind
                        {\x\mapsto\X}
                        {\AppExp{(\Rebind{c_1}{\X_1\mapsto\x})}{\Rebind{c_2}{\X_2\mapsto\x}
                        } }
                      }
        }
$$
The term $f_1$ represents a meta-operator for combining two different components $c_1$ and $c_2$
that have to be ``connected'' through the two input names $\X_1$ and $\X_2$, respectively.
The output of the component composition specified by $f_1$ is a new component with
just one input name $\X$ connected to both $\X_1$ and $\X_2$, and thus identifying the
two names of $c_1$ and $c_2$ (see \refToFigure{boxes}, left-hand-side).

The following term $f_2$ can be used for managing the opposite situation where the same input name of two different components $c_1$ and $c_2$
has to be discriminated when combining the two components (see \refToFigure{boxes}, right-hand-side).

$$
f_2 = \LambdaExp{c_1}
          {
            \LambdaExp{c_2}
                      { \Unbind
                        {\y_1\mapsto\X_1, \y_2\mapsto\X_2}
                        {\AppExp{(\Rebind{c_1}{\X\mapsto\y_1})}{\Rebind{c_2}{\X\mapsto\y_2}
                        } }
                      }
        }
$$ 
 
\begin{figure}
\begin{center}
\includegraphics[keepaspectratio,width=0.6\textwidth]{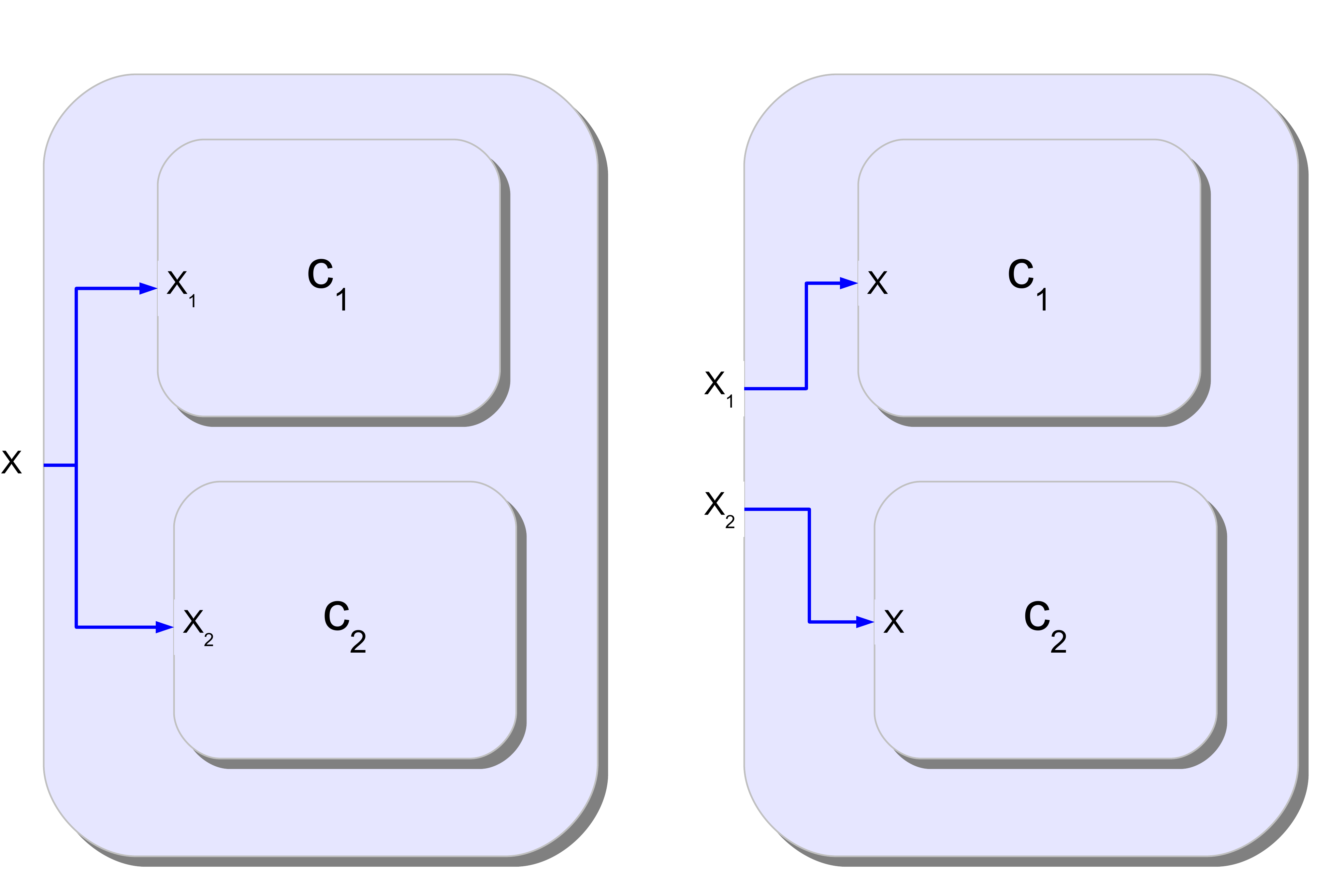}
\end{center}
\caption{Component composition with name identification (left) or discrimination (right).}\label{fig:boxes}
\end{figure}    

Operators corresponding to the functions $f_1$ and $f_2$, as defined above, can be expressed in the $\mathit{MML}^N_\nu$ calculus of \cite{AnconaMoggi04}.
In  $\mathit{MML}^N_\nu$, modules can be expressed as suitable combinations of records and code fragments; the code fragment expression $b(r)e$ is the analogous of
the unbound term $\Unbind{\tsubstN}{\te}$, where $r$ is a binding variable, occurring free
in $e$, which is expected to be dynamically bound to a resolver, that is, a map from names to expressions (indeed, resolvers correspond
to the rebinding maps); differently from the calculus presented in this paper, input names in code fragments are not referenced by means
of variables and unbinding maps but by means of resolver variables and the dot notation. 

The term $\Unbind{\x\mapsto \X,\y\mapsto \Y}{x+y}$ can be encoded in $\mathit{MML}^N_\nu$ by the term $b(r)r.X+r.Y$,
while the term  $\AppExp{(\LambdaCode{\z}{\X\mapsto 1,\Y \mapsto 2}{\z})}{\Unbind{x\mapsto X,y\mapsto Y}{x+y}}$ can be encoded by
the term $(b(r)r.X+r.Y)\langle?\{X{:}1\}\{Y{:}2\}\rangle$.
In $\mathit{MML}^N_\nu$ the expression $e\langle\theta\rangle$ allows the linking of the code fragment $e$ with the resolver
$\theta$ (the resolver $?\{X{:}1,Y{:}2\}$ corresponds to
the rebinding map $\X\mapsto 1,\Y\mapsto 2$). 

Differently from the calculus presented here, 
in $\mathit{MML}^N_\nu$ linking of code fragments $e\langle\theta\rangle$ is kept distinct from standard function application. However, 
$\mathit{MML}^N_\nu$ supports features not covered by the calculus presented in this paper: fresh names generation and multi-stage programming \cite{MetaML} 
(thanks to computational types).

The features of our calculus support meta-programming ``in the large'', promoting dynamic composition and reconfiguration of software components;
other interesting and finer-grained kinds of meta-programming, like multi-stage programming \cite{MetaML} and first-class patterns \cite{Jay04}, 
are beyond the scope of the calculus and would require non trivial extensions to be supported.  

\section{Conclusion}

\EZ{We have presented a minimal calculus which smoothly 
integrates positional  and nominal binding.}

\EZ{Despite its simplicity, this calculus provides a unifying foundation for module 
composition/adaptation, meta-programming, mobile code, 
and dynamic binding of variables. }\EZComm{Ho spostato le cose che erano related work nella sezione precedente}

%{Furthermore, as shown in Example~\ref{ex:meta-prog}, the calculus supports meta-programming features, by allowing programmers to define functions
%able to manipulate and combine code fragments.} 

\EZ{Soundness can be guaranteed by a type system where types are hierarchical, that is, an unbound type $\TypeWithMod{\xtenv}{\T}$ is the type of open code,
where $\xtenv$ describes the types of names to be rebound, and $\T$ can be an unbound type in turn}. These types have a modal interpretation studied in \cite{NanevskiEtAl08}. However, in our calculus we may have free variables in unbound terms (that could be bound in lambda-abstractions later on). Therefore, we may have terms of type  $\funType{\T}{\codeType{}{\T}}$ and $\funType{\codeType{}{\T}}{\T}$, implying that $\codeType{}{\T}$ and $\T$ would be, as modal formulas, equivalent, {thus making} the modal interpretation not correct. 
\EZ{We plan to investigate this interesting issue in further work}. 

\EZ{An alternative approach to guarantee soundness, that we described in previous work \cite{DezaniEtAl09,DezaniEtAl10} consists in simplifying the types so that they only take into account the number of rebindings needed to obtain a ground type, and to combine static and dynamic type checking.
That is, rebinding raises a dynamic error if for some variable there is no replacing term or it has the wrong type.}

As explained at the end of Section \ref{subS:compOur}, particular rebinding lambda-expression, which are values,  may be interpreted as representing substitutions, i.e., contexts of execution, moreover at the beginning of Section \ref{subS:compOther} we showed how the presence of free variables in unbound terms allows us to express ``unquote''. We are investigating how to increase these meta-programming capabilities of our calculus to support the code manipulation required by multi-stage programming, see \cite{MetaML}.

\bibliographystyle{eptcs}
{\small
\bibliography{Main}}

\end{document}